\documentclass[peerreview,draftclsnofoot,onecolumn]{IEEEtran}
\usepackage{caption}

\usepackage{cite}

\usepackage{subcaption}
\usepackage[utf8]{inputenc}

\usepackage{amssymb}
\usepackage{amsfonts}
\usepackage{amsmath}
\usepackage{graphicx}
\usepackage{indentfirst}
\usepackage{dsfont}

\usepackage{mathtools}

\usepackage{enumerate}
\usepackage{url}
\usepackage{array}
\usepackage{multirow}
\usepackage{longtable}

\usepackage{amsthm}

\newtheorem{definition}{Definition}
\newtheorem{theorem}{Theorem}
\newtheorem{corollary}[theorem]{Corollary}

\newtheorem{lemma}[theorem]{Lemma}

\newtheorem{remark}{Remark}

\usepackage[symbol]{footmisc}

\usepackage{tikz}
\definecolor{darkblue}{rgb}{0.15,0.35,0.55}
\definecolor{reddish}{rgb}{.8, 0.2, 0.2}
\usepackage[pdfpagelabels, linktocpage=true]{hyperref}
\hypersetup{
colorlinks=true,
citecolor=darkblue,
linkcolor=reddish,
urlcolor=darkblue,
pdfauthor={},
pdftitle={},
pdfsubject={}
}

\newcommand{\EX}{\mathbf{E}}

\usepackage[T1]{fontenc}

\long\def\ca#1\cb{} %Use for commenting out: \ca...\cb

\newcommand{\becs}{\begin{cases}}
\newcommand{\bem}{\begin{matrix}}

\newcommand{\dya}[1]{|#1\rangle\langle#1|}
\newcommand{\dyad}[2]{|#1\rangle\langle#2|}

\newcommand{\encs}{\end{cases}}
\newcommand{\enm}{\end{matrix}}

\newcommand{\ket}[1]{|#1\rangle }

\newcommand{\mat}[1]{\left(\begin{matrix}#1\end{matrix}\right)}

\newcommand{\ot}{\otimes }
\newcommand{\Tr}{{\rm Tr}}

\newcommand{\AC}{{\mathcal A}}
\newcommand{\BC}{{\mathcal B}}

\newcommand{\DC}{{\mathcal D}}
\newcommand{\EC}{{\mathcal E}}

\newcommand{\HC}{{\mathcal H}}
\newcommand{\IC}{{\mathcal I}}

\newcommand{\LC}{{\mathcal L}}

\newcommand{\NC}{{\mathcal N}}

\newcommand{\XC}{{\mathcal X}}
\newcommand{\YC}{{\mathcal Y}}

\newcommand{\rB}{\textbf{r}}

\newcommand{\al}{\alpha }
\newcommand{\bt}{\beta }
\newcommand{\gm}{\gamma }

\newcommand{\Dl}{\Delta }

\newcommand{\lm}{\lambda }
\newcommand{\Lm}{\Lambda }

\newcommand{\sg}{\sigma }

\DeclareMathOperator{\tr}{tr}

\usepackage[normalem]{ulem}

\DeclareMathAlphabet{\pazocal}{OMS}{zplm}{m}{n}

\title{Unital Qubit Queue-channels: Classical Capacity and Product Decoding}
\author{Vikesh Siddhu,  Avhishek Chatterjee, Krishna Jagannathan, Prabha Mandayam,  Sridhar
    Tayur%
\thanks{V. Siddhu was with JILA, University of Colorado/NIST, 440 UCB, Boulder, CO 80309, USA (email: vsiddhu@protonmail.com). He is presently with IBM Quantum, IBM T.J. Watson Research Center, New York.}
 \thanks{A.~Chatterjee and K.~Jagannathan are with the Department of Electrical Engineering, IIT Madras, Chennai, India (e-mail: \{avhishek, krishnaj\}@ee.iitm.ac.in).}%
\thanks{P. Mandayam is with the Department of Physics, IIT Madras, Chennai, India (e-mail: prabhamd@iitm.ac.in).}% <-this % stops a space
\thanks{S. Tayur is with the Quantum Computing Group, Tepper School of Business, and Department of ECE, Carnegie Mellon University, Pittsburgh PA 15213, USA (email: stayur@cmu.edu).}}

\begin{document}
\maketitle

\begin{abstract}
    Quantum queue-channels arise naturally in the context of buffering in
    quantum networks, wherein the noise suffered by the quantum states depends
    on the time spent waiting in the buffer. It has been shown that the
    upper-bound on the classical capacity of an additive queue-channel has a
    simple expression, and is achievable for the erasure and depolarizing
    channels \cite{MandayamJagannathanEA20}.  In this paper, we characterise
    the classical capacity for the class of unital qubit queue-channels, and
    show that  a simple product (non-entangled) decoding strategy is
    capacity-achieving. As an intermediate result, we  derive an explicit
    capacity achieving product decoding strategy for any i.i.d. unital  qubit
    channel, which could be of independent interest. As an important special
    case, we also derive the capacity and optimal decoding strategies for a
    symmetric generalized amplitude damping (GAD) queue-channel. Our results
    provide useful insights towards designing practical quantum communication
    networks, and highlight the need to explicitly model the impact of
    buffering.
\end{abstract}

\section{Introduction}
\label{sec:intro}
There is considerable and growing interest in designing and setting up
large-scale quantum communication networks \cite{PirandolaBraunstein16,
wehner2020}. To that end, understanding the fundamental capacity limits of
quantum communications in the presence of noise is of practical importance. In
this context, the inevitable buffering of quantum states during communication
tasks acts as an additional source of decoherence. One concrete example of such
buffering occurs at intermediate nodes or quantum repeaters, where quantum
states have to be stored for a certain \emph{waiting time} until they are
processed and transmitted again~\cite{nemoto2016}. Indeed, while quantum states
wait in buffer for transmission, they continue to interact with the
environment, and suffer a \emph{waiting time dependent}
decoherence~\cite{repeater2018, repeater_waitingtime20}. In fact, the longer a
qubit waits in a buffer, the more it decoheres.

To characterise the impact of buffering on quantum communication, researchers
have recently combined queuing models with quantum noise models
\cite{MandayamJagannathanEA20}. In particular, the buffering process inherently
introduces correlations across the noise process experienced by consecutive
qubits, since the waiting times are correlated according to the queuing
dynamics. Thus, to properly characterise the decoherence introduced due to
buffering, we need to look `beyond i.i.d' quantum channels and noise models.

Unital qubit channels are ubiquitous models for
decoherence~\cite{NielsenChuang11, Holevo12} in the communication medium as
well as  in the buffer. Though the former mode of decoherence has been the main
topic of interest in quantum Shannon theory, recent research has started to
focus on the impact of the buffering on the design of a practical quantum
communication system \cite{repeater2018,
repeater_waitingtime20,nemoto2016,MandayamJagannathanEA20}.

The i.i.d. unital channel has been studied extensively and its classical
capacity has been characterized~\cite{King02, DattaRuskai05, FukudaGour17,
GillardBelinEA19, King14, MendlWolf09}. The classical capacity is known to be
additive and is achieved by non-entangled (product) encoding.  However, to the
best of our knowledge, the following questions have not been resolved: (a) can
product \emph{decoding} achieve the classical capacity of the channel, and (b)
if so, is there an explicit quantum measurement that achieves the capacity?
These questions are well-motivated regardless of any buffering considerations,
because entangled measurement (non-product decoding) requires a reliable
quantum processor. Motivated by the practical issue of decoherence during
buffering, we further ask: what is the impact of decoherence at the
transmission buffer on the classical capacity and does it change the answers to
questions (a) and (b)?

\subsection{Related Work}
\label{SrelWork}

Our work interleaves different aspects of quantum communication networks, from
quantum Shannon theory to queuing theory. In quantum Shannon theory, one
studies ultimate limits for transmitting information in the presence of quantum
noise. One simple model of study is transmission of classical information
across qubits experiencing i.i.d noise. However even this simple model can
exhibit a variety of complex behaviour~\cite{BennettFuchsEA97, Fuchs97,
KingNathansonEA02}.
The qubit generalized amplitude damping channel~(GADC) is a relevant model of
noise in a variety of physical contexts including communication over optical
fibers or free space~\cite{YuenShapiro78, Shapiro09,
ZouLiEA17,RozpedekGoodenoughEA18}, $T_1$ relaxation due to coupling of spins
with a high temperature environment~\cite{ChuangNielsen97,
MyattKingEA00,TurchetteMyattEA00}, and super-conducting based quantum
computing~\cite{ChirolliBurkard08}. Quantum capacities of the i.i.d. GADC have
been studied~(see~\cite{KhatriSharmaEA20} and reference therein). Of particular
interest to us are expressions for the Holevo information of the GADC, found
in~\cite{LiZhenMaoFa07} using techniques from~\cite{Cortese02, Berry05}, and
channel parameters~\cite{KhatriSharmaEA20} where additivity of the GADC Holevo
information is known.
While the primary focus of quantum Shannon theory~\cite{Wilde19} has been to
study the classical and quantum capacities of stationary, memoryless quantum
channels~\cite{Holevo12}, recently there has been a spurt of activity in
characterizing the capacities of quantum channels in non-stationary, correlated
settings. We refer to~\cite{RMP2014} for a recent review of the different
capacity results obtained in a context of quantum channels that are not
independent or identical across channel uses. In particular, we focus on the
quantum information-spectrum approach in~\cite{HayashiNagaoka03}, which
provides bounds on the classical capacity of a general, non-i.i.d. sequence of
quantum channels. 
The idea of a quantum queue-channel was originally proposed
in~\cite{qubits_ncc2019} as a way to model and study the effect of decoherence
due to buffering or queuing in quantum information processing tasks. The
classical capacity of quantum queue-channels has been studied for certain
classes of quantum channels, and a general upper bound is known for additive
quantum queue-channels, additionally the upper bound can be achieved for for
the \emph{erasure} and \emph{depolarising}
queue-channels~\cite{MandayamJagannathanEA20}. The effect of queuing-dependent
errors on classical channels has been studied earlier~\cite{ChatterjeeSeoEA17},
with motivation drawn from crowd-sourcing.  More recently, a dynamic
programming based framework for characterising the queuing delay of quantum
data with finite memory size has been proposed in~\cite{qdelay_2020}. Finally,
we note that ideas of queuing theory have also been used to study aspects of
entanglement distribution over quantum networks such as
routing~\cite{EntRouting_2019}, switching, and buffering~\cite{guha2021}.

\subsection{Our Contributions:}

We show that the upper-bound on the classical capacity of additive
queue-channel is achievable for any unital qubit queue-channel if the encoder
has non-causal side information regarding the waiting times of the qubits. In
the absence of this side information, we show that for the class of unital
qubit queue-channels that are `Pauli-ordered,' the same upper-bound can be
achieved. In both cases, we show that non-entangled projective measurements can
achieve the capacity and provide explicit descriptions of the encoders and the
projective measurements.  As an intermediate result, we derive a capacity
achieving non-entangled projective measurement for any i.i.d. unital qubit
channel. To the best of our knowledge, this result has not been discussed in
the literature, and could be of independent interest.

An important example of a unital channel is the symmetric generalized amplitude
damping (GAD) channel. The GADC, $\AC_{p,n}$, is typically parametrized by two
quantities, $n$ and $p$, both between zero and one. At $n=1/2$, one obtains a
symmetric GADC which is unital, however for other values of $n$ the channel is
not unital. We construct three different 'natural' induced channels and find
one of them, $\NC_3$ a binary symmetric channel, to have the largest Shannon
capacity $C(\NC_3)$ for all $n$ and $p$. Typically, $C(\NC_3)$ is found to be
less than the GAD channel's Holevo information, however at $n=1/2$ we find
$C(\NC_3)$ equals the  Holevo information for all $p$.  Next, we study a
symmetric GAD queue-channel with $n=1/2,$ and the parameter $p$ is made an
explicit function of the waiting time $w$ of each qubit. Such a symmetric GAD
queue channel is unital and hence additive, which enables the use of the
capacity upper bound obtained in \cite{MandayamJagannathanEA20}. Further, we
propose a specific encoding for the GAD queue channel, which induces a binary
symmetric classical queue channel. We show that an achievable rate of this
binary symmetric queue channel matches the upper bound enforced by additivity
arguments, thus settling the capacity of the GAD queue channel, and giving us a
fully classical capacity achieving scheme for the encoder and decoder. Finally,
we obtain useful insights for designing practical quantum communication systems
by employing queuing theoretic analysis on the queue-channel capacity results.

\begin{figure}[ht]
    \centering
    \includegraphics[width=16 cm, height = 3.8 cm]{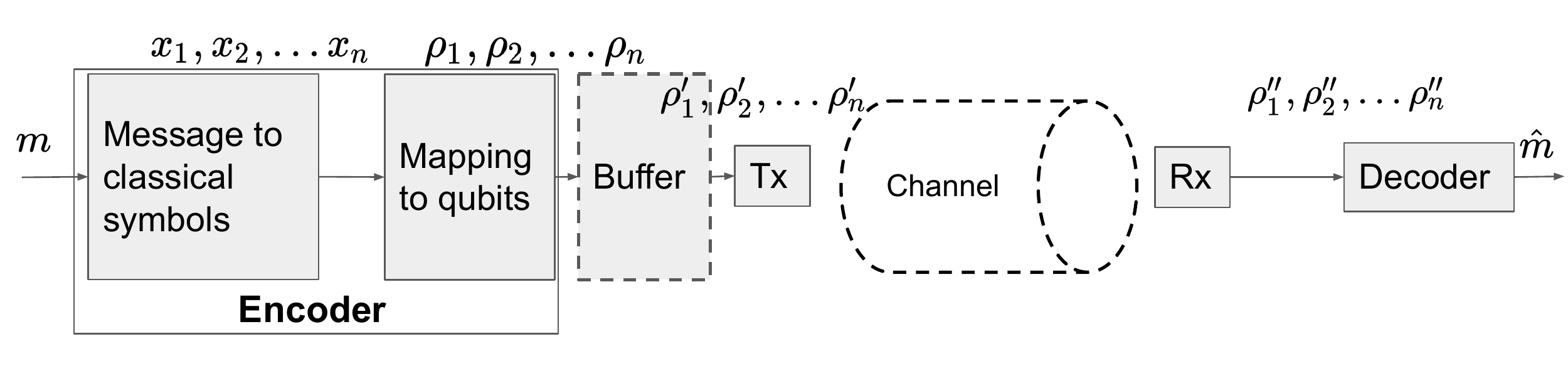}
    \caption{Qubit $\rho_i$ decoheres to $\rho_i'$ while waiting in the buffer
    for transmission. This further decoheres to $\rho''_i$ while passing
    through the channel. Decoherence in the buffer depends on the waiting time
    and results in non-i.i.d. "effective" decoherence.}
    \label{fig:schematic}
\end{figure}
%

%%%%%%%%%%%%%%%%%%%%%%%%%%%%%%%%%%%%%%%%%%%%%%%%%%%%%%%%%%%%%%%%%%%%%%%%%%%%%%%%
%%%%%%%%%%%%%%%%%%%%%%%%%% PAPER ORGANIZATION %%%%%%%%%%%%%%%%%%%%%%%%%%%%%%%%%%
%%%%%%%%%%%%%%%%%%%%%%%%%%%%%%%%%%%%%%%%%%%%%%%%%%%%%%%%%%%%%%%%%%%%%%%%%%%%%%%%

The paper is organized as follows. In the Sec.~\ref{SrelWork} we discuss
related work. To keep this discussion somewhat self-contained, in
Sec.~\ref{SPrelim}, we provide an extended discussion of induced channels,
classical capacities of quantum channels, and non-i.i.d queue-channel
capacities. 
Sec.~\ref{UnitalQueue} discusses unital qubit queue-channels and includes a
capacity achieving product encoding-decoding strategy for i.i.d. unital
channels~(see Th.~\ref{thm:unitalPOVM}).
In Sec.~\ref{SAmpDamp}, we analyze the generalized amplitude damping
channel~(GADC).  Here we discuss and compare capacities of various natural
choices for induced channels of a GADC~(see Fig.~\ref{Fig_ICHolevo}).  In
Sec.~\ref{sec:GADqueue} we discuss the queue-channel capacity of the symmetric
GADC.  We offer useful design insights by analyzing and numerically
plotting~(see Fig.~\ref{fig:LambdaVsCap}) the capacity expression.
Sec.~\ref{sec:cncl} contains a brief discussion and outlines potentially
interesting future directions.

\section{Preliminaries}
\label{SPrelim}

\subsection{Classical and Quantum Channels}

A random variable, $X$, taking discrete value $x$ from a finite set $\XC$ with
$p(x) := \Pr(X=x)$ has Shannon entropy $H(X)= -\sum_{x \in \XC} p(x) \log_2
p(x)$. A discrete memoryless channel $N$ taking $x \in \XC$ to $y \in \YC$ with
conditional probability $p(y|x):= \Pr(Y=y|X=x)$ has channel capacity
\begin{equation}
    C_{\text{Shan}}(N) = \max_{p(x)} I(X;Y),
    \label{ch1_ChanCap2}
\end{equation}
where $I(X;Y) := H(X) + H(Y) - H(X,Y)$ is the mutual information between input
$X$ and output $Y$.
A binary symmetric channel~(BSC) with flip probability $q$, is defined by the
conditional probability distribution $p(0|0) = 1- q, p(1|0) = q, p(0|1) = q,$
$p(1|1) = 1-q$, where $0 \leq q \leq 1$; it has capacity $1 - h(q)$, where
$h(q) = -[q \log_2 q + (1-q) \log_2 (1-q) ]$ is the binary entropy function. A
binary asymmetric channel~(BAC), defined by the conditional probability
distribution $p(0|0) = 1- q, p(1|0) = q, p(0|1) = r,$ and $p(1|1) = 1-r$, has
capacity \textcolor{black}{
\begin{equation}
    C_{\text{Shan}}\big (\text{BAC}(q,r) \big) = \frac{1}{1-s-t}(s h(t) - (1-t)h(s)) + \log_2 (1 + 2^{\frac{h(s) - h(t)}{1-s-t}}),
    \label{eq:BACCap}
\end{equation}
} where $s = \min(q,r)$ and $t = \max(q,r)$.

Let $\HC$ denote a finite dimensional Hilbert space, $\LC(\HC)$ denote the
space of bounded linear operators on $\HC$. A density operator $\rho$ is a
positive semi-definite operator in $\LC(\HC)$ with unit trace, $\Tr(\rho) = 1$.
A classical quantum~(c-q) channel, $\EC: \XC \mapsto \LC(\HC)$ maps a symbol $x
\in \XC$ to a density operator $\rho(x) \in \LC(\HC)$. 
Measuring $\rho(x)$ and recording the measurement outcome $y \in \LC(\YC)$
can be represented by a map $\DC : \LC(\HC) \mapsto \YC$.
This measurement can be described using a POVM, a collection of positive
operators in $\LC(\HC)$ that sum to the identity. Suppose the POVM
$\{\Lm(y)\}$ specifies $\DC$; then any input $x \in \XC$ is decoded as $y$ with
conditional probability,
\begin{equation}
    p(y|x) = \Pr(Y = y | X = x) = \Tr\big( \Lm(y) \rho(x) \big).
    \label{cond1}
\end{equation}
This conditional probability defines an {\em induced channel} $N : \XC \mapsto
\YC$ with capacity $C_{\text{Shan}}(N)$~(see Ch.20 in~\cite{Wilde19}). For a
fixed $\EC$, maximizing this capacity over choice of decodings $\DC$ defines
the Shannon capacity of $\EC$,
\begin{equation}
    C_{\text{Shan}}(\EC) = 
    \max_{\DC} C_{\text{Shan}}(N) = 
    \max_{\DC, p(x)} I (Y;X).
\end{equation}
For a fixed $p(x)$ and output alphabet $\YC$, $I(X;Y)$ is convex in $p(y|x)$,
and $p(y|x)$ is linear in the decoding POVM $\{\Lm(y)\}$ specifying $\DC$.
The resulting convexity of $I(X;Y)$ in $\DC$, for fixed $\YC$, is one reason
due to which the capacity $C_{\text{Shan}}(\EC)$ is non-trivial to compute.
A larger capacity can be obtained using a fixed product encoding $\EC: \XC \mapsto \LC(\HC)$
by allowing for decoding $\DC_k : \LC(\HC^{\ot k}) \mapsto \YC^{\ot k}$ that
may jointly measure $k$ encoded states. Such decoding $\DC_k$ defines
an induced channel $N_k$.
Maximizing the channel mutual information $\IC(N_k)$ over all decodings $\DC_k$
defines $\IC^{(k)}(\EC)$.
Due to the presence of entanglement in the joint decoding measurements, one may
have $\IC^{(k)}(\EC) \geq k\IC^{(1)}(\EC)$. Due to this type of
{\em super-additivity}, a proper definition of the capacity $C_{pj}$ of sending
classical information using {\em p}roduct encoding $\EC$ and {\em j}oint
decoding is given by a {\em multi-letter} formula,
\begin{equation}
    C_{pj} = \lim_{k \mapsto \infty} \frac{1}{n}\IC^{(k)}(\EC).
    \label{HolevoChiOp}
\end{equation}
Remarkably, the Holevo-Schumacher-Westmoreland theorem~\cite{Holevo98,
SchumacherWestmoreland97} gives the above multi-letter expression a {\em
single-letter} form; that is,
\begin{equation}
    C_{pj}(\EC) = \chi^{(1)}(\EC) := \max_{\{p(x)\}} \chi \big( p(x), \rho(x) \big),
    \label{HolevoChi1}
\end{equation}
where the Holevo quantity,
\begin{equation}
    \chi\big( p(x), \rho(x) \big) = S\big(\sum_{x \in \XC} p(x) \rho(x) \big) 
    - \sum_{x} p(x) S\big(\rho(x))\big),
    \label{HolevoChi}
\end{equation}
and $S(\rho) = -\Tr(\rho \log \rho)$, is the von-Neumann entropy of a density
operator $\rho$. Due to the close connection between
$C_{pj}$ and $\chi$, sometimes $C_{pj}(\EC)$ is also denoted by
$C_{\chi}(\EC)$.  There are cases where $C_{\chi}(\EC)$ is strictly greater
than $C_{\text{Shan}}(\EC)$~\cite{SasakiBarnettEA99, Shor04}.  However, much
remains unknown about when and how such separations \textcolor{black}{occur}. 

\subsection{Classical Capacities of a Quantum Channel}
\label{sec:CCap}

A quantum channel $\BC : \LC(\HC_a) \mapsto \LC(\HC_b)$ is a completely
positive trace preserving~(CPTP) map.
The Shannon capacity of $\BC$,
\begin{equation}
    C_{Shan}(\BC) = \max_{\EC, \DC} \IC(N) = \max_{ \{\EC, \DC \} } \max_{p(x)} I(X;Y),
    \label{ShannonCap}
\end{equation}
where $N:X \mapsto Y$ is an induced channel obtained by using product encoding
$\EC: \XC \mapsto \LC(\HC_a)$ and product decoding $\DC: \LC(\HC_b) \mapsto
\YC$.
\textcolor{black}{Allowing for joint encodings while restricting the decoder to product decodings does not
increase the channel's ability to send classical information beyond the channel's Shannon capacity, i.e., $C_{jp}(\BC) =
C_{Shan}(\BC)$~\cite{KingRuskai_2001}}.
The Shannon capacity is bounded from above by $C_{pj}(\BC)$, sometimes called
the Holevo capacity or the product state capacity of $\BC$. A multi-letter
expression of the form~\eqref{HolevoChiOp} for $C_{pj}(\BC)$ can be shown to
equal to a single-letter formula,
\begin{equation}
    C_{pj}(\BC)= \chi^{(1)}(\BC) := \max_{\{\rho_a(x),p(x)\}} \chi \big(p(x), \rho_b(x) \big),
    \label{HolevoBChi}
\end{equation}
where $\rho_b(x) = \BC(\rho_a(x))$.
The most general capacity of a channel $C_{jj}(\BC)$, sometimes called the {\em
classical capacity} of $\BC$ allows for joint encoding $\EC_k:\XC^{\times k}
\mapsto \LC(\HC_a^{\ot k})$ and decoding $\DC_k: \LC(\HC^{\ot k}_b) \mapsto
\YC^{\times k}$. This encoding-decoding results in an induced channel $\tilde
N_k: X^{\times k} \mapsto Y^{\times k}$.  Maximizing the mutual information of
this induced channel over all $\EC_k$ and $\DC_k$ gives $\tilde \IC^{k}(\BC)$.
This quantity can be super-additive, as a result $C_{jj}(\BC)$ is defined by a
multi-letter expression of the form~\eqref{HolevoChiOp}.  Using the product
state capacity $\chi^{(1)}(\BC)$~\eqref{HolevoBChi}, $C_{jj}(\BC)$ can be
written as follows,
\begin{equation}
   C_{jj}(\BC) = \underset{k \mapsto \infty}{\lim} \frac{1}{k} \chi^{(1)}(\BC^{\ot k}) := \chi(\BC).
   \label{ch1_HolCap}  
\end{equation}
In general, the limit in~\eqref{ch1_HolCap} is required because the product
state capacity can be non-additive~\cite{Hastings09}; that is, for any two
quantum channels $\BC$ and $\BC'$, the inequality,
\begin{equation}
   \chi^{(1)}(\BC \ot \BC') \geq  \chi^{(1)}(\BC) + \chi^{(1)}(\BC'), 
   \label{ch1_HolNon}  
\end{equation}
can be strict. For certain special classes of channels, the Holevo information
is known to be additive; that is, the inequality above becomes an equality when
$\BC'$ is any channel and $\BC$ belongs to a special class of channels that
includes unital qubit channels~\cite{King02}, depolarizing
channels~\cite{King03}, Hadamard channels~\cite{KingMatsumotoEA07}, and
entanglement breaking channels~\cite{Shor02}.

\subsection{Classical capacity of non-i.i.d. quantum channels}
Much of the focus in quantum shannon theory is on quantum channels that are
{\em independent and identically distributed~(i.i.d.)} across multiple uses.
As mentioned in Sec.~\ref{sec:intro}, the effective channel seen by qubits in
the presence of decoherence  in the transmission buffer is non-i.i.d.
Characterizing the capacity is a harder problem in such a setting. In the
classical setting, a capacity formula for this general non-i.i.d. setting was
obtained using the information-spectrum method~\cite{Han03, VerduHan94}. This
technique was adapted to the quantum setting in~\cite{HayashiNagaoka03}, and a
general capacity formula was obtained for the classical capacity of a quantum
channel. 

\subsubsection{The Quantum inf-information rate}
Recall that a quantum channel is defined as a completely positive,
trace-preserving map $\BC : \HC_a \mapsto \HC_b$ from the "input" Hilbert space
$\HC_a$ to the "output" Hilbert space $\HC_b$. Consider a sequence of quantum
channels $\vec{\mathcal{N}} \equiv \{\mathcal{N}^{(n)}\}_{n=1}^{\infty}$. Let
$\vec{P}$ denote the totality of sequences $\{P^n( X^n ) \}_{n=1}^{\infty}$ of
probability distributions (with finite support) over input sequences $X^n$, and
$\vec{\rho}$ denote the sequences of states $\rho_{X^n}$ corresponding to the
encoding $X^n \rightarrow \rho_{X^n}$. For any $a \in \mathbb{R}^{+}$ and $n$,
we define the operator, \[ \mathcal{O}_{\{P^{n}(X^{n}), \rho_{X^{n}}\}}(a) =
\mathcal{N}^{(n)} (\rho_{X^{n}})- e^{an}\sum_{X^{n}\in
\mathcal{X}^{(n)}}P^{n}(X^{n})\mathcal{N}^{(n)}(\rho_{X^{n}}) . \] Further, let
$\{  \mathcal{O}_{\{P^{n}(X^{n}), \rho_{X^{n}}\}}(a)  > 0\}$ denote the projector
onto the positive eigenspace of the operator $\mathcal{O}_{\{P^{n}(X^{n}),
\rho_{X^{n}}\}}(a) $. 

\begin{definition}
The quantum inf-information rate~\cite{HayashiNagaoka03}
    $\underline{\mathbf{I}}( \, \{ \vec{P}, \vec{\rho} \, \}, \vec{\mathcal{N}}
    \, )$  is defined as,
\begin{equation}
	\underline{\mathbf{I}}( \, \{ \vec{P}, \vec{\rho} \, \}, \vec{\mathcal{N}} \, ) = 
	\sup \left\lbrace a\in\mathbb R^+\left\vert \lim_{n\rightarrow \infty} \sum_{X^{n}\in \mathcal{X}^{(n)}} P^{n}(X^{n}) 
    \tr \left[ \mathcal{N}^{(n)} (\rho_{X^{n}}) \left\lbrace \mathcal{O}_{\{P^{n}(X^{n}), \rho_{X^{n}}\}}(a)   > 0 \right.\right\rbrace \right] = 1 \right\rbrace .
    \label{eq:quantum_I}
	\end{equation}
\end{definition}

This is the quantum analogue of the classical inf-information rate originally
defined in~\cite{Han03, VerduHan94}. The central result
of~\cite{HayashiNagaoka03} is to show that the classical capacity of the
channel sequence $\vec{\mathcal{N}}$ is given by \[C = \sup_{\{\vec{P},
\vec{\rho}\}}\underline{\mathbf{I}} (\{\vec{P},\vec{\rho}\},\vec{\mathcal{N}}).
\]
 
\section{Unital Qubit Queue-channels}
\label{UnitalQueue}

A unital qubit channel $\Phi$ satisfies $\Phi(I) = I$ where $I$ is the $2
\times 2$ identity operator. By itself, the channel describes i.i.d. noise.
The capacity of sending classical information in this i.i.d. setting was
discussed in Sec.~\ref{sec:CCap}, where we mentioned that the product state
classical capacity of $\Phi$ is additive and thus the channel's capacity,
$\chi(\Phi)$, can be achieved using product encoding.

{\color{black} A unital qubit queue-channel models the total decoherence experienced by the qubits while  waiting in the buffer for transmission and passing through the channel. Each qubit experiences a (potentially)
different unital qubit channel $\Phi_W$ parametrized by the random time $W$ that it
spends in the buffer.} In this case, for the transmitted
state $\rho_{12\ldots k}$, the output state would be $(\Phi_{W_1} \otimes
\Phi_{W_2} \cdots \Phi_{W_k})~(\rho_{12\ldots k})$, if the waiting times
$W^k=(W_1, W_2, \ldots, W_k)$ are known at the receiver. 
Examples of unital qubit queue-channels include the depolarising
queue-channels~\cite{MandayamJagannathanEA20} and the symmetric generalized
amplitude damping queue-channel~(see Sec.~\ref{CapUnQue}).

The buffering process is modeled as a continuous-time single-server queue. To
be specific, the single-server queue is characterised by (i) A server that
processes the qubits in the order in which they arrive, that is in a First Come
First Served (FCFS) fashion\footnote{The FCFS assumption is not required for
our results to hold, but it helps the exposition.}, and (ii) An "unlimited
buffer" --- that is, there is no limit on the number of qubits that can  wait
to be transmitted. We denote the time between preparation of the $i$th and
$i+1$th qubits by $A_i$, where $A_i$ are i.i.d. random variables. These $A_i$s
are viewed as inter-arrival times of a point process of rate $\lambda,$ where
$\mathbb E[A_i]=1/\lambda.$ The "service time," or the time taken to transmit
qubit $i,$ is denoted by $S_i$, where $\{S_i\}$ are also assumed to be i.i.d.
random variables, independent of the inter-arrival times $A_i,i\geq 1$. The
"service rate" of the qubits is denoted by $\mu=1/\mathbb E[S_i].$ We assume
that $\lambda<\mu$ (i.e., mean transmission time is strictly less than the mean
preparation time) to ensure stability of the queue. Qubit $1$ has a waiting
time $W_1=S_1$. The waiting times of the other qubits can be obtained using the
well known Lindley's recursion: \[W_{i+1} = \max(W_i - A_i, 0) + S_{i+1}.\] In
queuing parlance, the above system describes a continuous-time $G/G/1$ queue.
Under mild conditions, the sequence $\{W_i\}$ for a stable $G/G/1$ queue is
\emph{ergodic,} and reaches a \emph{stationary distribution} $\pi.$ We assume
that the waiting times $\{W_i\}$ of the qubits are {\em available} at the
receiver during decoding. 

An important difference between the queue-channel introduced above and the
usual i.i.d. channels is that this channel is a part of continuous time
dynamics. Hence, the usual notion of capacity per {\em channel use} for i.i.d.
channels is not pertinent here.  As mentioned before, the above channel model
is closely related to quantum queue-channels studied in
\cite{MandayamJagannathanEA20}. So, we first do a short review of the notion of
capacity per {\em unit time} and some relevant capacity results in
\cite{MandayamJagannathanEA20}.

\subsection{Classical capacity of unital quantum queue-channels}
\label{unitalCCap}

\begin{definition}
\label{def:contAchievableRate}
A rate $R$ is called an achievable rate for a quantum queue-channel if there
exists a sequence of $(n, 2^{R T_n})$ quantum codes with probability of error
$P_e^{(n)} \to 0$ as $n\to \infty$ and $\mathbf{E}\left[\sum_{i=1}^{n-1}
    A_i+W_n\right] \le T_n$.
\end{definition}

\begin{definition}
\label{def:capacity}
The information capacity of the queue-channel is the supremum of all achievable
    rates for a given arrival and service process, and is measured in bits per
    unit time.
\end{definition}
Note that the information capacity of the queue-channel depends on the arrival
process, the service process, and the noise model.

As discussed in Sec.~\ref{sec:intro}, in this paper, we derive the capacity of
this channel and show that product encoding and product decoding achieve that
capacity. Towards this, an important intermediate step of (possibly)
independent interest is to design an explicit product encoding and product
decoding strategy for i.i.d. unital qubit channels. 

\subsection{Product Encoding/Decoding for i.i.d. \textcolor{black}{Unital} Qubit Channels}
\label{sec:inducedChannel}
It is well known that product encoding achieves the classical capacity of an
i.i.d. unital qubit channel, which is equal to the Holevo information
\cite{King02}. In this section, we show that product  decoding is sufficient to
achieve that capacity and provide an explicit capacity achieving product
encoding and decoding strategy. To the best of our knowledge, this explicit
result is not available in the current literature.

The classical capacity of an i.i.d. unital qubit channel \textcolor{black}{$\Phi$} is given by the Holevo
information \cite{King02}
\begin{equation}
    \chi(\Phi) = \sup_{p, \rho, \rho'} \left(S(\Phi(p\rho+(1-p)\rho')) - p S(\Phi(\rho)) - (1-p) S(\Phi(\rho')) \right). 
    \label{eq:HolevoUntial}
\end{equation}
\subsubsection{Product encoding and decoding}
\label{sec:productEncDec}
For a unital qubit channel $\Phi$, let 
\begin{equation}
    M_{\Phi}=\sup_{\rho} ||\Phi(\rho)||, \label{eq:MPhi}
\end{equation}
where $||\cdot||$ is the operator norm~(it equals the largest eigenvalue of
a density operator). 

{\color{black} We define $\mathcal{R}_{\Phi}$ to be the set  of states that achieves the supremum in
\eqref{eq:MPhi}, i.e., for any $\rho \in \mathcal{R}_{\Phi}$,  $M_{\Phi}= ||\Phi(\rho)||$.

For any state $\rho \in \mathcal{R}_{\Phi}$, we define $\Gamma_{\Phi,\rho}$ to be the set of states such that for any $\tau \in \Gamma_{\Phi,\rho}$ 
\begin{equation}
    M_{\Phi}=\text{Tr}(\Phi(\rho) \tau).
    \label{eq:tauStar}
\end{equation}

{\bf Message to classical bits:} Consider the classical binary symmetric channel (BSC) with cross-over
probability $M_{\Phi}$ and choose any capacity achieving encoder and decoder.
For example, one can choose the well known random coding and typical decoding,
or an  appropriate polar code and the corresponding decoder.

{\bf Encoding classical bits to quantum states:} For sending a message over the unital channel, first map the message to an
appropriate classical binary codeword from the chosen classical codebook. Then
map symbol $0$ to a state $\rho^* \in \mathcal{R}_{\Phi}$ and symbol $1$ to $I -\rho^*$, and transmit
over the unital channel.

{\bf Decoding quantum states:} At the receiver, use projection measurements $\{P_0=I-\tau^*, P_1=\tau^*\}$, where $\tau^* \in \Gamma_{\Phi,\rho^*}$, 
and obtain a sequence of $0$ and $1$. Then, use the classical decoder chosen
for the BSC$(M_{\Phi})$.
 }

\begin{theorem}
\label{thm:unitalPOVM}
The above product encoding and decoding strategy for the unital qubit channel
    achieves the capacity in \eqref{eq:HolevoUntial}.
\end{theorem}

\begin{proof}
    First, we prove that the above encoding and decoding across an i.i.d.
    unital qubit channel results in a classical i.i.d. BSC~($M_{\Phi}$). The
    rest follows using the fact that $\chi(\Phi)=1-h(M_{\Phi})$ \cite{King02}
    and $1-h(M_{\Phi})$ is the Shannon capacity of BSC$(M_{\Phi})$.

    The probability that bit $0$ is decoded as bit $1$ is equal to the
    probability that the projective measurement $\{P_0,P_1\}$ on $\Phi(\rho^*)$
    gives $1$. Similarly, the probability that bit $1$ is decoded as bit $0$ is
    same as the probability of the event that the projective measurement on
    $\Phi(I-\rho^*)$ gives $0$. 
    The second probability is given by
    \begin{align}
        & \ \ \text{Tr}((I-\tau^*) \Phi(I-\rho*) )\nonumber \\
        & = \text{Tr}((I-\tau^*)(I-\Phi(\rho^*))) \ \ \ \text{  (unital channel)} \nonumber \\
        & = \text{Tr}(I-\tau^*-\Phi(\rho^*)+\tau^* \Phi(\rho^*)) \nonumber \\
        & = \text{Tr}(\tau^* \Phi(\rho*)). \nonumber
    \end{align}
    This expression, however, is exactly equal to first probability, which in turn
    is given by
    \begin{align}
        & \ \ \text{Tr}(\tau^* \Phi(\rho*)) \nonumber \\
        & = \sup_{\tau: \text{pure state}} \text{Tr}(\Phi(\rho^*) \tau) \nonumber \\
        & = ||\Phi(\rho^*)|| \ \ \ \text{ (by the defn. of operator norm) } \nonumber \\
        & = M_{\Phi}. \nonumber
    \end{align}
    This completes the proof.
\end{proof}
The main insight from the above theorem is summarized in the following remark.
\begin{remark}
Every unital qubit channel has an induced binary symmetric channel whose
    Shannon capacity equals the classical capacity of the unital qubit channel.
\end{remark}

Next, building on the above insight and Theorem~4 in
\cite{MandayamJagannathanEA20}, we study unital qubit queue-channels.

\subsection{Capacity of Unital Qubit Queue-channels}
\label{CapUnQue}
We start with the capacity upper-bound in~\cite{MandayamJagannathanEA20}, which
is applicable to any additive queue-channel. We assume that the waiting times
of the qubits are {\em available} at the receiver during decoding.

\begin{theorem}[\cite{MandayamJagannathanEA20}, Theorem~1]
\label{thm:capUB}
    The classical capacity of a unital qubit queue-channel is upper-bounded by {\color{black}
    $\lambda~\EX_{\pi} \chi(\Phi_W)$, irrespective of whether the
    encoder knows the waiting times or does not know the waiting times. Here, $\mathbf{E}_{\pi}$ is expectation with respect to the stationary  distribution $\pi$ of $\{W_i\}$.}
\end{theorem}
\begin{proof}
    The case where waiting times are not known at the encoder is a direct
    re-statement from \cite{MandayamJagannathanEA20}. The other case 
    follows by noting the fact that in
    deriving the upper-bound in \cite{MandayamJagannathanEA20}, the encoder was
    allowed access to the additional side information regarding the waiting times.
\end{proof}

We study encoding and decoding strategies that achieve the above bound in both
settings. We start with the  simpler setting where encoder knows the waiting
times and later we study the more practical setting, where the encoder does not
know the waiting times.

\subsubsection{Encoder knows waiting times}
\label{sec:TxknowsW}
Knowledge of the future parameters of a time-varying channel at the receiver is
called non-causal side information.  This is not practical when the channel
variation is fast and unpredictable (i.i.d. like). However, as the waiting
times result into a Markov process, such an assumption is not so impractical.
In certain slowly varying queues, the waiting times can be predicted within a
reasonable accuracy. 
In this setting, the product encoding and product decoding strategy is similar
to the one considered in Sec.~\ref{sec:inducedChannel}. However, some
modifications are necessary to address the non-i.i.d. nature of the
queue-channel.

First, we introduce a modified version of \eqref{eq:HolevoUntial}. For an
unital qubit channel parametrized by waiting time $W$, let
{\color{black}
\begin{align}
& M_{\Phi_W}= \sup_{\rho} ||\Phi_W(\rho)||, \nonumber \\
   & \mathcal{R}_{\Phi_W}= \{\rho: ||\Phi_W(\rho)|| = M_{\Phi_W}\}, \nonumber \\
   & \Gamma_{\Phi_W,\rho} = \{\tau: \text{Tr}(\tau \Phi(\rho))= M_{\Phi_W}\} \text{ for } \rho \in \mathcal{R}_{\Phi_W}. \label{eq:OptEncDecW}
\end{align}
}

\noindent{\bf Message to classical bits:} Pick any capacity achieving encoder
and decoder for the classical binary symmetric queue-channel
$\{\text{BSC}(M_{\Phi_{W_i}})\}$. A detailed discussion on this channel can be
found in~\cite{MandayamJagannathanEA20}.

\noindent {\bf Product encoding and decoding of qubits:} {\color{black} The encoder and the decoder agree  a priori on a choice of $\rho_W \in \mathcal{R}_{\Phi_W}$ for all $W \ge 0$. The encoder maps the $i$th
classical bit to $\rho_{W_i}$ or $I-\rho_{W_i}$, depending on whether it is
$0$ or $1$, respectively. For the $i$th quantum state at the output of the channel, the decoder uses the projective measurement
$\{I-\tau_{W_i}, \tau_{W_i}\}$, where $\tau_{W_i} \in \Gamma_{\Phi_{W_i},\rho_{W_i}}$.}

\begin{theorem}
\label{thm:TxknowsW}
The above product encoding and product decoding strategy for unital qubit
    queue-channel achieves the capacity upper-bound in Theorem~\ref{thm:capUB}.
\end{theorem}
\begin{proof}
Using the steps from the proof of Theorem~\ref{thm:unitalPOVM}, it directly
    follows that  the above strategy converts the unital qubit queue-channel
    into a binary symmetric queue-channel $\{\text{BSC}(M_{\Phi_{W_i}})\}$.
    Rest follows from Theorem~4 in \cite{MandayamJagannathanEA20}.
\end{proof}

\subsubsection{Encoder does not know waiting times}
\label{TxNotknowW}
In this setting the queue evolution cannot be predicted and
hence, the encoder has no knowledge of $\{W_i\}$. This is a more prevalent
setting in quantum communication. In many practical quantum
communication systems, the encoding and the decoding has to be chosen at time
zero, and cannot be adapted according to the queue evolution. We show that, in
this setting, again a simple product encoder and product decoder achieves
capacity for a large class of \textcolor{black}{unital} qubit queue-channels.

{\color{black}We obtain two results in this setting. First, we show that for a class of unital qubit queue-channels with certain Pauli decomposition characteristics whose Pauli decompositions satisfy a certain invariant ordering, the capacity can be achieved by product encoding and decoding in terms of Pauli matrices (Theorem~\ref{thm:TxNotknowW} and Lemma~\ref{lem:maxpxpluspi}). This class of channels includes the well known depolarizing channels, the symmetric generalized amplitude damping channel and other Pauli channels such as bit-flip and phase-flip channels. Second, we further introduce a broader class of unital qubit queue-channels which can be characterized without using the Pauli decomposition} for which product encoding and decoding is optimal, independent of their Pauli noise characteristics (Theorem~\ref{thm:TIE}).

Let us consider a {\em family} of i.i.d. unital qubit channels $\{\Phi_w\}$,
parametrized by a non-negative real number $w$. This means that the channel
acts on any joint state $\rho_{12\ldots k}$ as \[(\Phi_w \otimes \Phi_w \otimes
\ldots \Phi_w)(\rho_{12 \ldots k}),\] where the parameter $w$ determines the
map. {\color{black} As discussed in Theorem~\ref{thm:unitalPOVM}, the classical capacity of
this channel is achieved by the product encoding: $0 \to \rho_w^*$ and $1 \to I-\rho^*_w$, and product decoding using the projectors $\{\tau^*_w, I-\tau^*\}$, where
\begin{align}
    \rho_w^* \in \mathcal{R}_{\Phi_w} \text{ and } \nonumber
   \tau^*_w \in \Gamma_{\Phi_w,\rho_w^*}. 
    \label{eq:OptEncDecW}
\end{align}
Further, the classical capacity of this unital qubit channel is equal to the Shannon capacity of a binary symmetric channel with crossover probability $M_{\Phi_w}=||\Phi_w(\rho_w^*)||$.}

It is well known that any qubit state $\rho$ can be expressed as linear combination
of the Pauli matrices $\sigma_i$, $i=1,2,3$, and $I$. This leads to three
natural induced classical channels for any qubit channel: map $0$ and $1$, respectively, to
$\frac{I+\sigma_i}{2}$ and $\frac{I-\sigma_i}{2}$, the projectors onto the two
eigenvectors of $\sg_i$, and measure using these same projectors.
{\color{black} For the i.i.d. unital qubit channel $\Phi_w$, parametrized by $w$, this leads
to three induced binary symmetric channels $B_i(w)$, $i=1, 2, 3$. 

A family of i.i.d. unital qubit channels $\{\Phi_w\}$ and an unital qubit
queue-channel are closely related. In a unital qubit queue-channel, each qubit
$i$ sees a different unital qubit channel depending on its waiting time $W_i$.
Thus, any unital qubit queue-channel can be described using a family of i.i.d. unital qubit  channels $\{\Phi_w\}$, such
that the channel seen by any qubit $i$ is $\Phi_w$, where $w=W_i$. Clearly, the physical environment of the buffer decides the nature of the queue-channel and thus, determines the 
family $\{\Phi_w\}$ that corresponds to it.

\begin{definition}
\label{timeOrderedUnital}
    We call a unital qubit queue-channel {\bf Pauli-ordered} if the ordering of
    the Shannon capacities of the induced channels $B_1(w)$, $B_2(w)$ and
    $B_3(w)$ of the corresponding family of the i.i.d. unital qubit channels $\{\Phi_w\}$ is the same for all $w \ge 0$.
\end{definition}

Examples of Pauli-ordered unital qubit queue-channels are depolarising
queue-channels~\cite{MandayamJagannathanEA20} and symmetric generalized
amplitude damping channels considered in Sec.~\ref{SAmpDamp}.

As discussed before, a queue-channel models decoherence of a qubit  due to its interaction with the environment while waiting in a buffer. In this context, it is physically well motivated to work in a Markovian regime, leading to the well known quantum Markov semigroup structure \cite{breuer2002} for the channel that models the decoherence.  It can be shown that for a unital qubit queue-channel with a Markov semigroup structure, the Shannon capacities of $B_1(w)$, $B_2(w)$ and $B_3(w)$ do not change with $w$. Thus,  a unital qubit queue-channel with a Markov semigroup structure is indeed Pauli-ordered. This implies that the class of Pauli-ordered unital qubit queue-channels is a physically interesting and broad class of unital qubit queue-channels.

The following lemma is useful in designing an optimal encoding and decoding for  Pauli-ordered unital qubit queue channels.

\begin{lemma}
\label{lem:tauWSame}
    For a Pauli-ordered \textcolor{black}{unital} qubit queue-channel there exists a Pauli state $\hat{\sigma}$ such that for all $w\ge 0$, $(I+\hat{\sigma})/2 \in \mathcal{R}_{\Phi_w}$. This, in turn, implies that for any $w \ge 0$, $(I+\hat{\sigma})/2 \in \Gamma_{\Phi_w,(I+\hat{\sigma})/2}$.
\end{lemma}

Proof of this lemma is presented later. Here, we first derive an optimal product encoding and decoding strategy using this lemma.

{\bf Encoding and decoding:} We pick a capacity achieving encoder and decoder for
the classical binary symmetric queue-channel $\{\text{BSC}(M_{\Phi_{W_i}})\}$. We map the message to a string of $0$ and $1$ using that encoder.
Then, we map $0$ to $(I+\hat{\sigma})/2$ and $1$ to
$(I-\hat{\sigma})/2$ and use the projective measurement $\{(I-\hat{\sigma})/2, (I+\hat{\sigma})/2\}$ on the output states to obtain strings of $0$ and $1$. Finally, we use the capacity achieving decoder for the classical binary symmetric queue-channel $\{\text{BSC}(M_{\Phi_{W_i}})\}$ to decode the message.

\begin{theorem}
\label{thm:TxNotknowW}
    The above product encoding and product decoding strategy achieves the
    capacity upper-bound in Theorem~\ref{thm:capUB} for Pauli-ordered unital
    qubit queue-channels.
\end{theorem}
\begin{proof}
    Lemma~\ref{lem:tauWSame} implies that the above product encoding and decoding strategy
    converts a Pauli-ordered unital qubit queue-channel into a binary symmetric
    queue-channel with crossover probabilities 
    \[\text{Tr}\big(\Phi_{W_i}\left((I+\hat{\sigma})/2\right)~
    (I+\hat{\sigma})/2\big)=M_{\Phi_{W_i}}.\] 
    This is because by lemma~\ref{lem:tauWSame}, $(I+\hat{\sigma})/2 \in \mathcal{R}_{\Phi_{W_i}}$ and $(I+\hat{\sigma})/2 \in \Gamma_{\Phi_{W_i},(I+\hat{\sigma})/2}$.
    The rest follows from Theorem~4 in~\cite{MandayamJagannathanEA20}.
\end{proof}}

\begin{proof}[Proof of Lemma~\ref{lem:tauWSame}]
    First, note that up to local unitaries at the channel input and output, any
    unital channel can be written as a convex combination of Pauli channels :
    \[ \Phi(\rho) = (1-\sum_{i=1}^3 p_i) \rho + \sum_{i=1}^3 p_i \sigma_i \rho
    \sigma_i,\] 
    where $\sigma_i$ are the Pauli matrices~(see discussion between Prop.~6.41
    and Ex. 6.43 in~\cite{Holevo12}). 

    Thus, any $\Phi_w$ can be equivalently represented by the three
    probabilities $\{p_i(w), i=1,2,3\}$, where 
    \[ \Phi_w(\rho) = (1-\sum_{i=1}^3 p_i(w)) \rho + \sum_{i=1}^3 p_i(w)
    \sigma_i \rho \sigma_i,\] 
    and $\sum_i p_i(w) \le 1$.

    {\color{black} We prove Lemma~\ref{lem:tauWSame} using the following lemma, which gives
    an explicit expression for the optimal encoding and decoding in terms Pauli matrices.

\begin{lemma}
\label{lem:maxpxpluspi}
    For a unital channel $\Phi$, given by $\Phi(\rho)= (1-\sum_{i=1}^3 p_i) \rho +
    \sum_{i=1}^3 p_i \sigma_i \rho \sigma_i$, 
    \[ (I + \sigma_{i^*})/2 \in \mathcal{R}_{\Phi} \text{ and }  (I + \sigma_{i^*})/2 \in \Gamma_{\Phi,(I + \sigma_{i^*})/2},\] 
    where 
    \[i^* = \arg \max_{i\in \{1, 2, 3\}} |1-2\sum_{j\neq i}^3 p_j|.\]
\end{lemma}
}
    Lemma~\ref{lem:maxpxpluspi} is applicable to any $\Phi_w$ parametrized by
    $\{p_i(w)\}$. Thus, if the ordering of $\eta_i(w)=|\sum_{j=1}^3 p_j(w) -
    p_i(W) - \frac{1}{2}|$ remains unchanged with $w$, $\rho^*_w$ and $\tau^*_w$
    remain unchanged with $w$.
 
    Finally, note that the ordering of $\{\eta_i(w)\}$ are the same as the
    ordering of the Shannon capacities of $B_i(w)$.  To see this, let us first
    find the crossover probability of $B_i(w)$. Using simple trace
    calculations, one can show that the crossover
    probability $b_i(w)$ of the BSC $B_i(w)$ is $p_i + (1-\sum_{j=1}^3 p_j)$.
    Thus, its Shannon capacity is $1-h(b_i(w))$. Note that $1-h(b)=0$ at
    $b=\frac{1}{2}$ and increases monotonically with $|b-\frac{1}{2}|$. Thus,
    the capacity of $B_i(w)$ is monotonic in $\eta_i(w)$ and hence, ordering of
    $\eta_i(w)$ remains unchanged with $w$ for a Pauli-ordered unital qubit
    queue-channel, which completes the proof of this lemma.
\end{proof}

\begin{proof}[Proof of Lemma~\ref{lem:maxpxpluspi}]

    For any state $\rho=\frac{I}{2}+\sum_{i=1}^3 \frac{\alpha_i}{2} \sigma_i$,
    where $\sum_{i=1}^3 \alpha_i^2 \le 1$, 
    \[\Phi(\rho)=\sum_{i=1}^3
    \frac{\alpha_i}{2}(1-2\sum_{j\neq i}^3 p_j) \sigma_i+ (1-\sum_{i=1}^3 p_i)
    \frac{I}{2}.\] 
    Thus, $\rho= \frac{1}{2}(I + \vec{\alpha}.\vec{\sigma}) :=
    \frac{1}{2}(I+\sum_{i=1}^3 \alpha_i \sigma_i)$, where $\vec{\alpha} =
    (\al_1, \al_2, \al_3)$ and $\vec{\alpha}.\vec{\alpha} \le 1$, the channel
    output 
    \[\Phi(\rho)= \frac{1}{2}(I + (\vec{\lm}*\vec{\alpha}).\vec{\sigma}),  \]
    where $\vec{\lm} = (\lm_1, \lm_2, \lm_3)$, $\lm_i = 1 - 2 \sum_{j \neq i}
    p_j$, and \textcolor{black}{$(\vec{\lm}*\vec{\alpha})_i =
    \vec{\lm}_i\vec{\alpha}_i$ denotes entry-wise dot product between
    $\vec{\lm}$ and $\vec{\al}$}. Thus, 
\begin{align} 
    M_{\Phi} & = \sup_{ \{ \vec{\al}: \; \vec{\al}.\vec{\al} \leq 1 \}, \; \text{pure} \;
    \tau} 
    \frac{1}{2}\Tr \{ \big( (I + (\vec{\lm}*\vec{\alpha}).\vec{\sigma} \big) \tau \}  \nonumber \\
    & = \sup_{ \{ \vec{\al}: \; \vec{\al}.\vec{\al} \leq 1 \}, \; \{ \vec{\bt}: \; \vec{\bt}.\vec{\bt} = 1\} }
    \frac{1}{4}\Tr \{ \big( (I + (\vec{\lm}*\vec{\alpha}).\vec{\sigma} \big) (I + \vec{\bt}.\vec{\sigma}) \}  \nonumber
\end{align}
    The last step follows from the fact that $\tau$ is a pure state.

    After doing the matrix products and using some linear algebra involving
    linearity of trace, and the facts that $\Tr(\sigma_i)=0$, $\sigma_i^2 =I$,
    and for $i \neq j$, $\sg_i \sg_j = - \sg_j \sg_i$, one obtains
\begin{align}
    M_{\Phi} = 
    \sup_{ \{ \vec{\al}: \; \vec{\al}.\vec{\al} \leq 1 \}, \; \{ \vec{\bt}: \; \vec{\bt}.\vec{\bt} = 1\} }
    \frac{1}{2} \left( 1+ \textcolor{black}{(\vec{\lm}*\vec{\al}).\vec{\bt}} \right). \nonumber
 \end{align}
    It {\color{black} follows from the Cauchy-Schwartz inequality} that the supremum is obtained when $\vec{\al} = \vec{\bt}$
    and $|\bt_i|=1$ for $i=i^* = {\rm argmax} \; |\lm_i|$. Using the definition
    of $\lm_i$, it follows that
    $i^* = {\rm argmax} \; |1 - 2 \sum_{j \neq i} p_j|$.

\end{proof}

The notion of Pauli-ordered unital qubit queue-channels is directly connected to
the binary i.i.d. classical channels induced by the Pauli matrices. This gives
a physical interpretation of the conditions under which the statement in
Theorem~\ref{thm:TxNotknowW} holds true.  {\color{black}However, a result like
Theorem~\ref{thm:TxNotknowW}, holds for a broader class of unital qubit channels, which can be characterized without using any reference to their Pauli decompositions.
\begin{definition}
\label{def:symGADCGeneralized}
    Let $\{\Phi_w\}$ be the class of i.i.d. unital qubit channels corresponding
    to the unital qubit queue-channels. We call the queue-channel to have a {\bf
    waiting-invariant norm maximizer} if \[\bigcap_{w\ge 0} \mathcal{R}_{\Phi_w}\neq \emptyset.\]
\end{definition}

This class of queue-channels includes the class of Pauli-ordered queue-channels since it follows directly from Lemma \ref{lem:maxpxpluspi} that there exists a Pauli state $\hat{\sigma}$ such that $(I+\hat{\sigma})/2 \in \bigcap_{w\ge 0} \mathcal{R}_{\Phi_w}$.

{\bf Encoding and decoding:} Let $\bar{\rho}$ be a state in $\bigcap_{w\ge 0} \mathcal{R}_{\Phi_w}$ and $\bar{\tau}_{W_i}$ be a state in $\Gamma_{\Phi_{W_i},\bar{\rho}}$.
We pick a capacity achieving code for the
queue-channel BSC~($M_{\Phi_{W_i}}$) and generate classical codes accordingly. Then, we map $\{0,1\}$ to $\bar{\rho}$ and
$I-\bar{\rho}$ and decode the output states using the POVM $\{I-\bar{\tau}_{W_i}, \bar{\tau}_{W_i}\}$. Clearly, the encoder does not depend on individual $W_i$s, but the decoder may. In that sense also, this strategy is a generalization of the strategy used for Pauli-ordered channels.

\begin{theorem}
\label{thm:TIE}
    The above product encoding and product decoding strategy achieves the
    capacity upper-bound in Theorem~\ref{thm:capUB} for unital qubit
    queue-channels with a waiting-invariant norm maximizer.
\end{theorem}
\begin{proof}
    It is enough to prove that the above encoding and decoding strategy converts a unital qubit queue-channel with waiting-invariant norm maximizer into a binary symmetric queue-channel $\{\text{BSC}(M_{\Phi_{W_i}})\}$. The rest follows from Theorem~4 in  \cite{MandayamJagannathanEA20}.
    
    The crossover probability for state $i$ under this induced channel is given by
    \[\text{Tr}\big(\Phi_{W_i}\left(\bar{\rho}\right)~
    \tau_{W_i}\big)=\text{Tr}\big(\Phi_{W_i}\left(\rho\right)~\tau\big),\] 
    for $\rho \in \mathcal{R}_{\Phi_{W_i}}$ and $\tau \in \Gamma_{\Phi_{W_i},\rho}$. Hence, this quantity is equal to $M_{\Phi_{W_i}}$.
\end{proof}

Clearly, Theorem~\ref{thm:TIE} is more general than Theorem~\ref{thm:TxNotknowW} as it is applicable to a broader class of unital qubit queue-channels. However, Lemma~\ref{lem:maxpxpluspi}, which is an intermediate result for Theorem~\ref{thm:TxNotknowW}, 
 gives a simple closed form encoder and decoder in terms of Pauli matrices. This is of independent interest as it applies to any i.i.d. unital qubit channel as well. Also, we did not come across any practical scenario which may lead to unital queue-channels that are not Pauli-ordered. 
}

\section{Qubit Generalized Amplitude Damping~(GAD) Queue-channels}
\subsection{I.I.D GAD}
\label{SAmpDamp}
The qubit generalized amplitude damping~(GAD) channel 
$\AC_{p,n} : \LC(\HC_a) \mapsto \LC(\HC_b)$ is a two parameter family of
channels where the parameters $p$ and $n$ are between zero and one. The channel
has a qubit input and qubit output---$d_a = d_b = 2$--- and its superoperator
has the form,
\begin{equation}
    \AC_{p,n}(\rho) = \sum_{i=0}^3 K_i \rho K_i^{\dag},
    \label{ampDampGlm}
\end{equation}
where
\begin{align}
    K_0 &=   \sqrt{1-n}(\dyad{0}{0} + \sqrt{1-p}\dyad{1}{1}), 
    &K_1 &=   \sqrt{p(1-n)}\dyad{0}{1}, \\
    K_2 &=   \sqrt{n}(\sqrt{1-p}\dyad{0}{0} + \dyad{1}{1}), \quad \text{and}
    &K_3 &=   \sqrt{pn}\dyad{1}{0}
    \label{ampDampGlmKr}
\end{align}
are Kraus operators. The GAD~\eqref{ampDampGlm} channel can also be expressed as
\begin{equation}
    \AC_{p,n} = (1-n)\AC_{p,0} + n \AC_{p,1}.
    \label{ampDampGlmcv}
\end{equation}
The above representation provides the following insightful interpretation.  The
parameter $n$ represents the mixing of $\AC_{p,0}$ with $\AC_{p,1}$, where each
channel $\AC_{p,i}$~($i = 0$ or $1$) is an amplitude damping channel that
favours the state $[i]$~(\textcolor{black}{here we use the notation $[\psi]$
for $\dya{\psi}$}) by keeping it fixed and maps the orthogonal state
$[1-i]$ to $[i]$ with damping probability $p$.
When $n = 1/2$, $\AC_{p,n}$ is unital and we get equal mixing of both
$\AC_{p,0}$ and $\AC_{p,1}$. This equal mixing represents noise where each
state $[i]$~($i = 0,1$) is mapped to itself with probability \textcolor{black}{$p/2$} and to $[1-i]$
with probability \textcolor{black}{$1-p/2$}; in other words, this $n=1/2$ noise treats both $[0]$
and $[1]$ identically. However, when $n$ is not half, the action of $\AC_{p,n}$
on $[0]$ is different from its action on $[1]$.  In particular, $[0]$ is mapped
to itself with probability $1-pn$ and to $[1]$ with probability $pn$, and $[1]$
is mapped to itself with probability $1-p(1-n)$ and to $[0]$ with probability
$p(1-n)$.

Any qubit density operator can be written in the Bloch parametrization,
\begin{equation}
    \rho(\rB) = \frac{1}{2} (I + \rB. \vec{\sg}) 
    := \frac{1}{2}(I + x \sg^x + y \sg^y + z \sg^z),
    \label{BlochRep}
\end{equation}
where the Bloch vector $\rB = (x,y,z)$ has norm $|\rB| \leq 1$,
\begin{equation}
    \sg^x = \mat{0 & 1 \\ 1 & 0}, \quad 
    \sg^y = \mat{0 & -i \\ i & 0}, \quad \text{and} \quad
    \sg^z = \mat{1 & 0 \\ 0 & -1}
\end{equation}
are the Pauli matrices, written in the standard basis $\{\ket{0}, \ket{1}\}$.
Using the Bloch parametrization, the entropy 
\begin{equation}
    S\big( \rho(\rB) \big) = h\big( (1 - |\rB|)/2 \big),
    \label{qBitEntro}
\end{equation}
where $h(x):= -[x \log x + (1-x) \log (1-x)]$ is the binary entropy function
and $|\rB| = \sqrt{\rB.\rB}$ is the norm of $\rB$.
An input density operator $\rho(\rB)$ is mapped by $\AC_{p,n}$ to an output
density operator with Bloch vector,
\begin{equation}
    \rB' = (\sqrt{1-p}x, \sqrt{1-p}y, (1-p)z + p(1-2n) ).
    \label{ampGlnBlochVec}
\end{equation}
The GADC is unital at $n=1/2$; that is, $\AC_{p,\frac{1}{2}}(I) = I$. The GADC is
entanglement breaking~\cite{King02, Shor02, KhatriSharmaEA20} when
\begin{equation}
    2(\sqrt{2} - 1) \leq p \leq 1  \quad \text{and} \quad
    \frac{1}{2}(1 - l(p)) \leq n \leq \frac{1}{2}(1 + l(p)),
    \label{GADCebt}
\end{equation}
where $l(p) = \sqrt{\frac{p^2 + 4p -4}{p^2}}$. The Holevo capacity of unital
qubit channels and entanglement breaking channels is additive; as a result,
when $n=1/2$ or when the values of parameters $p$ and $n$
satisfy~\eqref{GADCebt}, the Holevo information of the generalized amplitude
damping channel,$\chi^{(1)}(\AC_{p,n})$ , equals the classical capacity of the
channel, $\chi(\AC_{p,n})$. For other values of $p$ and $n$, the classical
capacity of the GADC is not known because for these parameter values, the
Holevo information of the channel is not known to be additive or non-additive.
The actual value of the Holevo information can be computed numerically.  Next,
we briefly discuss this numerical calculation for completeness.

\subsection{Holevo Information}
\label{SGADCHolevo}
Let $[\al_+]$ and $[\al_-]$ be projectors on states with Bloch vector
\begin{equation}
    \rB_+ = (\sqrt{1-z^2},0,z), \quad \text{and} \quad 
    \rB_- = (-\sqrt{1-z^2},0,z),
    \label{ensGAPD}
\end{equation}
respectively; here $-1 \leq z \leq 1$. Notice, $[\al_+]$ and $[\al_-]$ are not
orthogonal, except when $z = 0$. It has been shown~\cite{LiZhenMaoFa07} that the
Holevo information,
\begin{equation}
    \chi^{(1)}(\AC_{p,n}) = \max_{\{-1 \leq z \leq 1\}} \; 
    S\big(\AC_{p,n}(\sg)\big) - [S\big(\AC_{p,n}([\al_+])\big) +
    S\big(\AC_{p,n}([\al_-])\big)]/2 ,
    \label{ch5_HolChi}  
\end{equation}
where $\sg = ([\al_+] + [\al_-])/2$. In the above equation, the optimizing
$z$ has the value
\begin{equation}
    z^* =\frac {u - p(1-2n)}{1-p},
    \label{ch5_optZGADC}
\end{equation}
where $u$ comes from solving,
\begin{equation}
    \big(pu - p^2 (1-2n) - p(1-p)(1-2n) \big) f'(r^*) = -r^*(1-\gm) f'(u),
\end{equation}
with
\begin{align}
    f(x)  &:= (1 + x) \log_2 (1+x) + (1-x) \log_2 (1-x), \\
    f'(x) &=  \log_2 \Big( \frac{1+x}{1-x}\Big), \quad \text{and} \\ 
    r^*   &:= \sqrt{1 - p - \frac{\big(u - p(1-2n) \big)^2}{1-p} + u^2 }.
\end{align}
Using the value of $z^*$ in \eqref{ch5_optZGADC} gives,
\begin{equation}
    \chi^{(1)}(\AC_{p,n}) = \frac{1}{2} \big( f(r^*) - \log_2 (1-u^2) -uf'(u) \big).
    \label{GADCHolevo}
\end{equation}
Solving~\eqref{GADCebt} for $n \leq 1/2$ gives a range,
\begin{equation}
    p^* \leq p \leq 1,
    \label{EBTRange}
\end{equation}
where the GAD channel in entanglement breaking. Here the value,
\begin{equation}
    p^* = \max\Big( 2(\sqrt{2} - 1),  \frac{\sqrt{1 + 4n(1-n)}-1}{2n(1-n)} \Big).
\end{equation}
As indicated earlier, entanglement breaking channels have additive Holevo
capacity.  Thus, when $p$ satisfies \eqref{EBTRange}, the GAD channel has additive
Holevo capacity. While the Holevo information $\chi^{(1)}(\AC_{n,p})$ gives the
product state classical channel capacity, it doesn't give an explicit encoding
and decoding that achieves this capacity. In what follows, we construct
explicit encoding and decodings---in other words, we construct induced classical channels,
and compare the capacity of these channels to the product state classical
capacity $\chi^{(1)}(\AC_{n,p})$. For $n=1/2$, we find the optimal encoding
and decoding which achieves $\chi^{(1)}(\AC_{1/2,p})$ for all $0 \leq p \leq 1$.

\subsection{Induced Channels}
\label{SsubICAmpGln}
To obtain an induced channel for $\AC_{p,n}: \LC(\HC_a) \mapsto \LC(\HC_b)$ one
must choose an encoding and decoding. To choose an encoding, $\EC: \XC \mapsto
\LC(\HC_a)$, one fixes a set of input states $\{\rho(x)\}$.  To choose a
decoding, $\DC: \LC(\HC_b) \mapsto \YC$, one fixes an output measurement POVM
$\{\Lm(y)\}$. Together this encoding-decoding results in an induced channel
with conditional probability $p(y|x) = \Tr(\rho(x) \Lm(y))$.
A priori, there is no clear choice for these input states and output
measurement. However, the generalized qubit amplitude damping channel satisfies
an equation
\begin{equation}
    \AC_{p,n}\big(\sg^z_a \; \rho \; (\sg^z_a)^{\dag} \big)  = \sg^z_b \; \AC_{p,n}(\rho) \; (\sg^z_b)^{\dag},
    \label{ampCov}
\end{equation}
where the subscripts $a$ and $b$ on the Pauli operator $\sg^z$ signify the
space on which the operator acts. The above equation implies that the
generalized amplitude damping channel has a rotational symmetry around the
$z$-axis.
Using this rotational symmetry and the fact that $\AC_{p,n}$ is a qubit
input-output channel one may choose an encoding $\EC(x) = \rho(x)$
where $x = 0$ or $1$ and $\{\rho(x)\}$ are two orthogonal input states that
remain unchanged under the $\sg^z_a$ symmetry operations; that is,
$\rho(x) = [x]$.  To decode, one may apply a protocol
for correctly identifying a state chosen uniformly from a set of two known
states $\AC_{p,n}([0])$ and $\AC_{p,n}([1])$ with highest probability.
This protocol comes from the theory of quantum state
discrimination~\cite{Helstrom69}. It uses a POVM with two elements $\{E, I_b -
E\}$, where $E$ is a projector onto the space of positive eigenvalues of
$\AC_{p,n}([0]) - \AC_{p,n}([1])$. An unknown state, either $\AC_{p,n}([0])$ or
$\AC_{p,n}([1])$ with equal probability, is measured using the POVM. If the
outcome corresponding to $E$ occurs, the unknown state is guessed to be
$\AC_{p,n}([0])$; otherwise, the guess is $\AC_{p,n}([1])$. In the present
case, a simple calculation shows that $E = [0]$.

Encoding $\EC(x) = [x]$,~($x = 0,1$) and decoding based on the POVM $\{[0],
[1]\}$, coming from the state distrimination protocol outlined above, results
in an induced channel $\NC_1$. This channel is a BAC that flips $x=0$ to $y=1$
with probability $pn$ but flips $x=1$ to $y=0$ with probability $p(1-n)$. Its
capacity $C(\NC_1)$ has a simple closed form expression\textcolor{black}{~\eqref{eq:BACCap}}. For
$C(\NC_1)$, this expression is unchanged when $n$ is replaced with $1-n$, thus
we may restrict our attention to $0 \leq n \leq 1/2$. 

At $n = 1/2$, $\AC_{p,n}$ is unital. In Sec.~\ref{sec:inducedChannel}, we
defined an induced channel which achieves the Holevo information of any qubit
unital channel.  On the basis of that induced channel, we may construct an
induced channel for values of $n$ different from $1/2$. In this construction
the encoding map $\EC(0) = \rho^*$, and $\EC(1) = I-\rho^*$~($\rho^*$ defined below
eq.~\eqref{eq:MPhi}); the decoding map $\DC$ measures the output of
$\AC_{p,n}$ using the POVM $\{\tau^*, I-\tau^*\}$~($\tau^*$ defined in
eq.~\eqref{eq:tauStar}) to return $0$ when the measurement outcome corresponds to
POVM element $\tau^*$, otherwise return $1$. This encoding-decoding results in
the induced channel $\NC_2$ which is a BAC. This BAC flips input $0$ to output $1$ with
probability $(1 - |\rB'|)/2$, $|\rB'|^2 = 4n(1-n)(1-p) + (1-2n)^2$ and flips
input $1$ to output $0$ with probability $\frac{1}{2}(1 + \frac{p\big(1 - (1-2n)^2
\big)-1}{|\rB'|})$. The channel's capacity, $C(\NC_2)$~(computed using \textcolor{black}{expression
\eqref{eq:BACCap}}), remains invariant when $n$ is
replaced with $1-n$. This invariance permits us to restrict ourselves to the
parameter range $0 \leq n \leq 1/2$.

Next, we consider a third induced channel. This channel is based on the
computation of the Holevo information of $\AC_{p,n}$ in Sec.~\ref{SGADCHolevo}.
Here, encoding is performed using possibly non-orthogonal states and decoding
is performed using a measurement designed to distinguish these encoded states
at the channel output with maximum probability.  The encoding maps $x=0$ and
$x=1$ to $[\al_+]$ and $[\al_-]$~(defined via eq.~\eqref{ensGAPD}),
respectively. The decoding is performed using a POVM $\{ \Lm(y) \}$ where
$\Lm(y)$ at $y=0$ is the projector onto the space of positive eigenvalues of
$\AC_{p,n}([\al_+]) - \AC_{p,n}([\al_-])$. This projector is simply $[x+]$,
where \textcolor{black}{$\ket{x+} := (\ket{0} + \ket{1})/\sqrt{2}$}. This encoding-decoding scheme
results in a one-parameter family of induced channels $\NC_3(z)$. This channel is
a BSC with flip probability $q(z) = (1- a(z)\sqrt{1-p})/2$, where $a(z) =
\sqrt{1-z^2}$.  Interestingly, this family of induced channels, coming from the
two parameter GAD channel $\AC_{p,n}$, does not depend on the parameter $n$.
The Shannon capacity of $\NC_3(z)$ is simply
\begin{equation}
    C(\NC_3(z)) = \; 1 - h \big(q (z) \big).
    \label{M2Shannon}
\end{equation}
For a fixed $p$, one can easily show that $C(\NC_3(z))$ is maximum when $z = 0$;
thus, $\NC_3 =\NC_3(0)$ has the largest Shannon capacity among the one-parameter
family of induced channels $\NC_3(z)$. This induced channel $\NC_3$ is simply a
BSC with flip probability $q = (1 - \sqrt{1-p})/2$. 

We compare the capacities of the three induced channels $\NC_1$, $\NC_2$, and
$\NC_3$.  As mentioned earlier, we can restrict ourselves to $0 \leq n \leq 1/2$.
A straightforward calculation shows that at $n=0$, $\NC_1$ and $\NC_2$ are
equivalent up to permutation of the inputs and output and thus $C(\NC_1) =
C(\NC_2)$. In general, $0 \leq n \leq 1/2$, here simple numerics can be used to
show that
\begin{equation}
    C(\NC_1) \leq C(\NC_2) \leq C(\NC_3).
\end{equation}
All inequalities above are numerically found to be strict when $0 < n < 1/2$
and $0 < p < 1$. At $n=1/2$, $\NC_2$ and $\NC_3$ become identical, they are both
BSC with flip probability $q = (1 - \sqrt{1-p})/2$. This flip probability can
be easily shown to equal $\NC_{\AC_{p,1/2}}$~(defined in eq.~\eqref{eq:MPhi}).
Using this equality, or the fact that $\NC_2 = \NC_3$ is the induced channel which
achieves the Holevo information when $\AC_{p,n}$ is unital at $n=1/2$, we
conclude $C(\NC_3) = \chi^{(1)}(\AC_{p,1/2})$.

For values of $n<1/2$ we compare the capacity of the $\NC_3$, the induced channel
with the largest capacity among $\NC_1$, $\NC_2$, and $\NC_3$ with
$\chi^{(1)}(\AC_{p,n})$. We numerically find that for values of $n < 1/2$ and
$0 < p < 1$, $C(\NC_3)< \chi^{(1)}(\AC_{p,n})$~(see Fig.~\ref{Fig_ICHolevo}).

\begin{figure}[ht]
    \begin{center}
        \includegraphics[scale=.75]{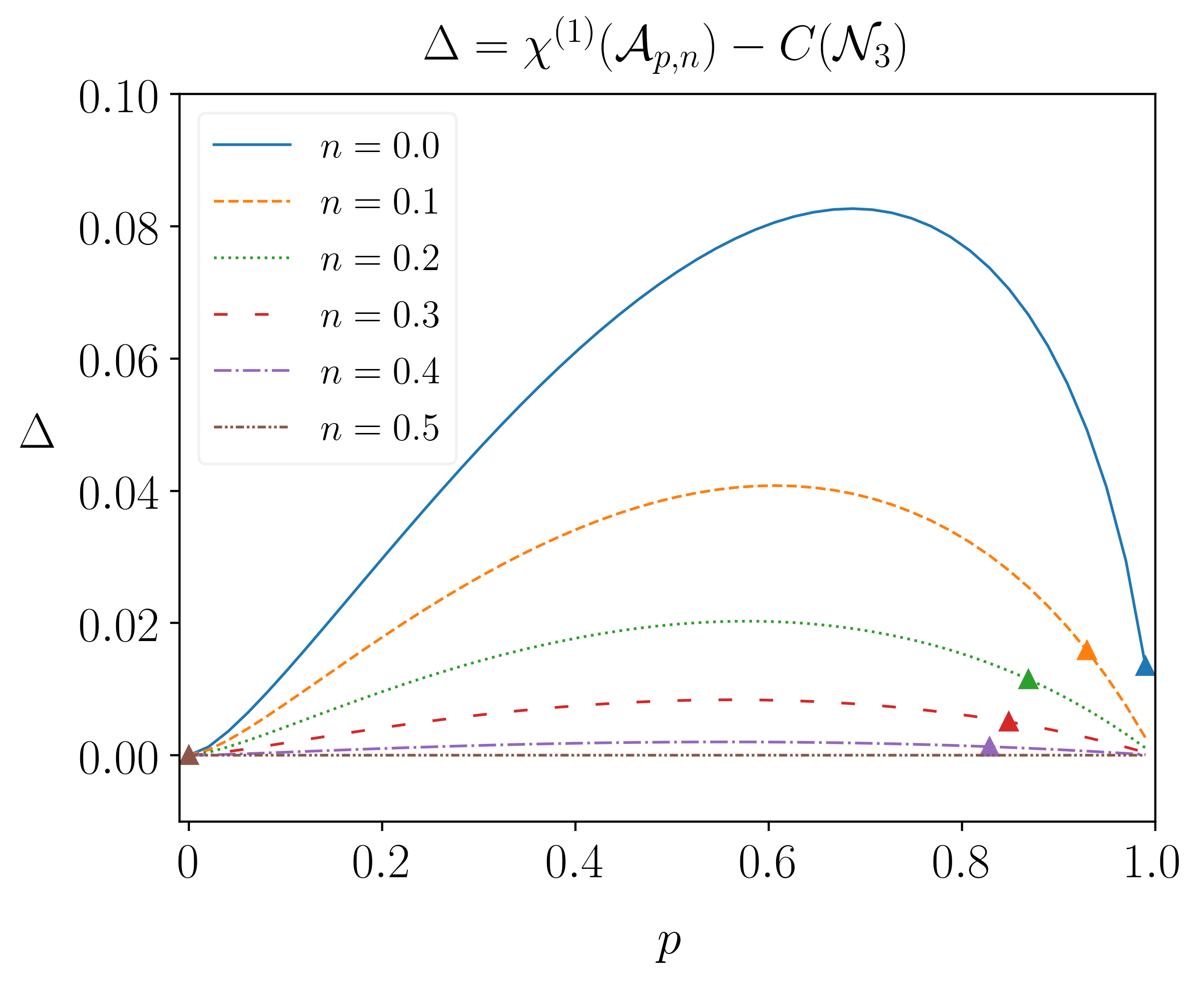}
        \caption{The difference $\Dl = \chi^{(1)}(\AC_{p,n}) - C(\NC_3)$ as a
        function of $p$ for various values of the parameter $n$. For each $n$,
        the colored triangle indicates the value of $p^*$ above which 
        $\chi^{(1)}(\AC_{p,n})$ is additive.\label{Fig_ICHolevo}}
    \end{center}
\end{figure}

In what follows, we focus on the $n=1/2$ GADC $\AC_{p,1/2}$. As discussed
below~\eqref{ampDampGlmcv}, this channel describes noise in which both
computational basis states $\ket{0}$ and $\ket{1}$ are treated on equal
footing. When information about which of these computational basis states
decays faster than the other, the GADC with $n\neq 1/2$ is an apt noise model.
However when such information is unavailable, or when it is known that both
computational basis states decay but the maximally mixed state doesn't, one
uses the $n=1/2$ GADC. One simple example of such noise is the qubit thermal
channel~(analogous to the bosonic thermal channel~\cite{KhatriSharmaEA20,
MyattKingEA00,TurchetteMyattEA00}) in which the channel environment is
represented by the maximally mixed state. Another simple example is the effect
of dissipation to an environment at a finite temperature~\cite{NielsenChuang11}.

\subsection{Capacity of the symmetric GAD queue-channel}
\label{sec:GADqueue}
For a symmetric GADC, the parameter $p$ captures the level of damping
experienced by a qubit while interacting with an environment. In the absence
of buffer decoherence, $p$ depends on the flight time $T_f$ through the
channel and the physical parameters of the channel. Similarly, the level of
damping experienced in the buffer depends on the waiting time in the buffer
$W$ and the physical parameters of the buffer. Hence, the effective GADC
parameter experienced by a qubit is a function $g(T_f,W)$ of its waiting time
and its flight time, where the form of $g(\cdot)$ depends on the physical
parameters of the channel and the buffer. As the flight time is almost
deterministic, for simplicity of notations we denote this function by
$p_{\mbox{eff}}(W)$. 

The capacity of a symmetric GAD queue-channel can be expressed as follows.
\begin{theorem}
\label{thm:qcAchieve}
The capacity of a symmetric GAD queue-channel is 
\[\lambda~\mathbf{E}_{\pi}\left[1-h\left(\frac{1-\sqrt{1-p_{\mbox{\em eff}}(W)}}{2}\right)\right].\]
\end{theorem}
\begin{proof}
As shown in Sec.~\ref{SsubICAmpGln}, the optimal encoding and the optimal POVM
    for symmetric GAD channels do not change with the channel parameter $p$.
    Thus, the symmetric GAD queue-channel allows a time-invariant encoding.
    Hence, Theorem \ref{thm:TIE} for unital qubit queue-channel with time
    invariant encoding is applicable to symmetric GAD queue-channels with
    parameter $p_{\mbox{eff}}(W)$. 

The rest follows by noting that the induced classical channel of a symmetric GAD
    channel with parameter $p$ is a binary symmetric channel with flip
    probability $(1-\sqrt{1-p})/2$.
\end{proof}
 
\subsection{Useful design insights} 
As the motivation for this work is the practical issues faced by current
quantum networks, we discuss few important practical insights obtained from the
analytical results for symmetric GAD queue-channels.
\begin{figure}
    \centering
    \includegraphics[scale=0.4]{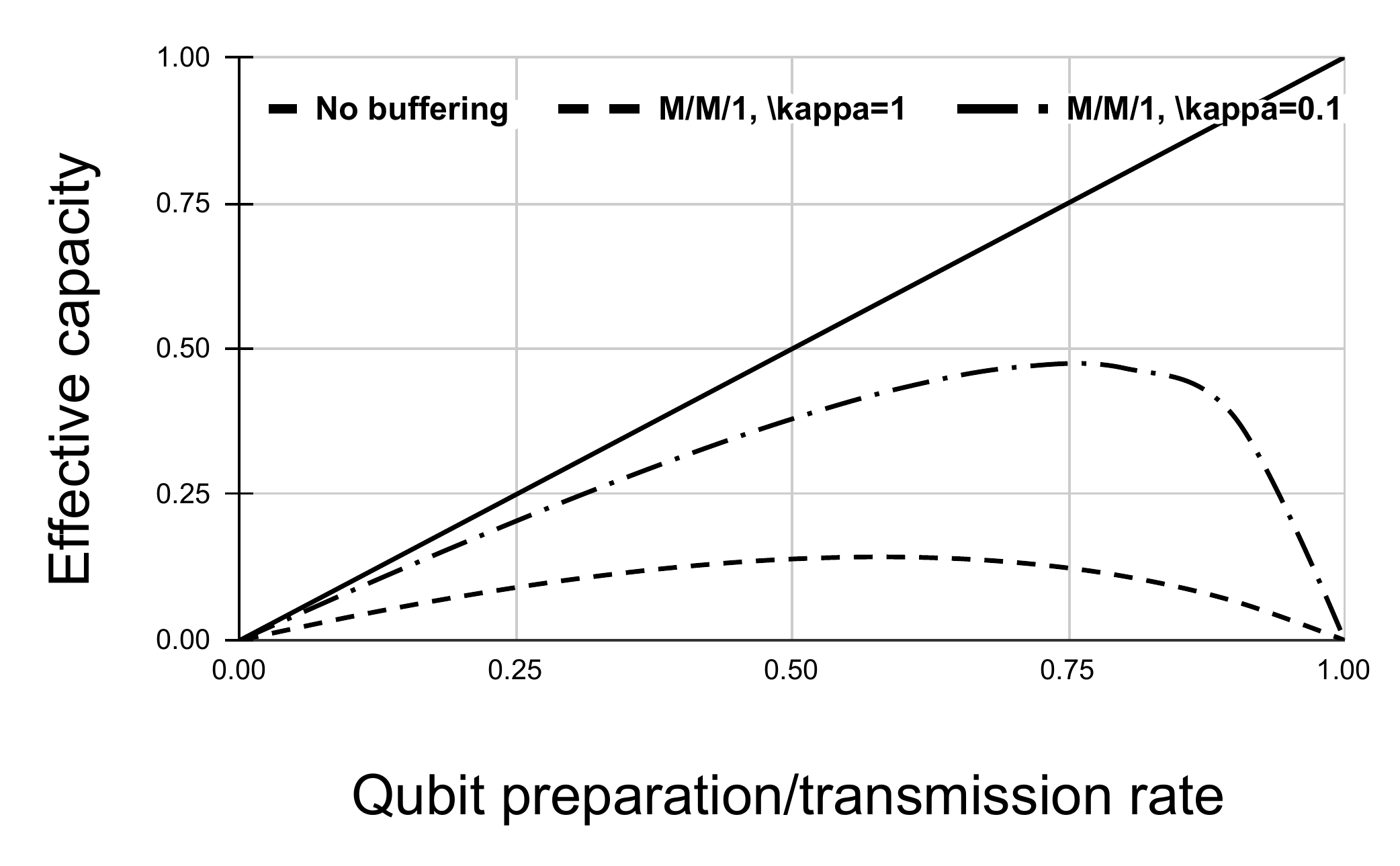}
    \caption{Capacity (effective) vs qubit preparation rate ($\lambda$) for different buffer decoherence.}
    \label{fig:LambdaVsCap}
\end{figure}
%%%%%%%%%%%

In Fig.~\ref{fig:LambdaVsCap}, the capacity per unit time (in contrast to per
channel use) of an idealized i.i.d symmetric GADC with $p=0$ is plotted (no
buffering) against the qubit preparation rate. This has the misleading
implication that \textcolor{black}{the} higher the qubit preparation rate, the
higher is the capacity.  However, it is well known that in any  practical
system, especially at a high qubit preparation rate, there will be significant
buffering at the transmitter, which will result in additional decoherence of
qubits, of significant magnitude, thus, resulting in the loss of capacity. This
is a fundamental concept in communication network design.

To illustrate this, we use a simple queue-channel model involving the well known M/M/1
queue \cite{Kleinrock75} that can analytically capture the loss in capacity at a
high qubit preparation rate due to buffering. In Fig.~\ref{fig:LambdaVsCap}, two such plots are
shown for symmetric GADC with M/M/1 buffering and exponential decoherence 
\begin{align}
    p_{\mbox{eff}}(W) = 1 - \exp(-\kappa~W),
    \label{eq:expDecoh}
\end{align}
 with mean decoherence time $\kappa^{-1}$. We obtain these plots using the capacity expression in Theorem~\ref{thm:capUB}. 
 
 Clearly, in Fig.~\ref{fig:LambdaVsCap}, the optimal $\lambda$ is not close to $\mu$ ($=1$). Moreover, for $\lambda$ close $\mu$, the capacity is almost zero. This is because very high
$\lambda$ leads to large waiting times for qubits and thus results in
significant decoherence. Furthermore, the optimal $\lambda$ depends on $\kappa$
and hence, on the physical parameters of the buffer. The idealized i.i.d.
setting fails to capture this crucial dependence.

In general, obtaining a closed form expression for the best $\lambda$ is not
possible. However, for any buffering discipline, the best $\lambda$ can be obtained by solving 
\[\arg\max_{\lambda\in(0,\mu)}
\lambda~\mathbf{E}_{\pi}\left[1-h\left(\frac{1-\sqrt{1-p_{\mbox{eff}}(W)}}{2}\right)\right].\]
Though it may appear that the capacity expression increases with $\lambda$, it
is not so since $\pi(\cdot)$ depends on $\lambda$. 

\subsection{Optimal queuing distributions}
\label{sec:queueOpt}
The effective capacity in the presence of buffer decoherence is a function of
the stationary distribution of waiting times. Thus, in turn, it is heavily
influenced by the time between preparation of two qubits and the time to
process (transmit and receive) a qubit. A quantitative understanding of this
dependence is useful for designing quantum communication systems. 

In this section, we take a short stride in that direction by characterizing the
optimal distributions in two queuing settings of general interest when the
channel and buffer decoherence follows the exponential model in Eq.
\ref{eq:expDecoh}.  The exponential decoherence model is physically the most
well motivated model for capturing decoherence in terms of the interaction time
with the environment.

First, we obtain a simpler expression of the capacity result in
Theorem~\ref{thm:qcAchieve} for the exponential decoherence model. 
\begin{corollary}
\label{cor:capSeries}
The effective capacity in the presence of buffer decoherence is given by
    \[\frac{\lambda}{\ln{2}} \sum_{k=1}^\infty \frac{1}{2k~(2k-1)}  \mathbb{E}_{W\sim \pi}\left[\exp\left(-\kappa~k~W\right)\right],\]
    when $p_{\mbox{\em eff}}(W)~=~1~-~\exp(-\kappa~W)$ for some $\kappa>0$.
\end{corollary}
\begin{proof}
For the exponential decoherence model, the capacity expression in Theorem~\ref{thm:qcAchieve} becomes
\[\lambda~\EX_{\pi} \left[1~-~h\left(\frac{1-\exp\left(-\frac{1}{2}\kappa~ W\right)}{2}\right)\right].\]
The rest follows using the series expansion of $\log(1+x)$ for $|x|<1$ and algebraic manipulations.
\end{proof}
Note that the expression in Cor.~\ref{cor:capSeries} is valid for any stable queue, irrespective of the queuing discipline and distributions. 
 
In the queuing literature, M/G/1 and G/M/1 are two popular classes of queuing
models. In our setting, M/G/1 is equivalent to exponentially distributed
(memoryless) preparation times and generally distributed processing or service
times of qubits. G/M/1 is equivalent to generally distributed preparation times
and exponentially distributed processing or service times. As a first step
towards optimizing queuing distributions, one may ask: what are the best
distribution for processing times and preparation times in M/G/1 and G/M/1
queues, respectively? The following theorems answer this question.

\begin{theorem}
\label{thm:MD1best}
Among all quantum communication systems  with M/G/1 buffering, symmetric GAD
    channel, and exponential decoherence, the system with deterministic
    processing or service time has the maximum effective capacity for any
    $\lambda$ and $\mu$ ($>\lambda$).
\end{theorem}
\begin{proof}
Suppose there exists a service distribution for which $\mathbb{E}_{W\sim
    \pi}\left[\exp(-s W)\right]$ is more than any other service distribution
    with the same mean for any $s>0$. Then, from the capacity expression in
    Corollary~\ref{cor:capSeries}, it is clear that under that particular
    distribution, each term in the series will be \textcolor{black}{greater than}
    the corresponding term for any other distribution. Hence, that distribution
    will achieve the maximum capacity among the class of all service
    distributions with the same mean.

Thus, to complete this proof, we need only to show that for exponentially
    distributed preparation times, the deterministic service time maximizes
    $\mathbb{E}_{W\sim \pi}\left[\exp(-s W)\right]$ for any $s>0$. This follows
    directly from the proof of Theorem~4 in \cite{MandayamJagannathanEA20}.
\end{proof}

\begin{theorem}
\label{thm:DM1best}
Among all quantum communication systems  with G/M/1 buffering, symmetric GAD
    channel, and exponential decoherence, the system with deterministic
    preparation/arrival time has the maximum effective capacity for any
    $\lambda$ and $\mu$ ($>\lambda$).
\end{theorem}
\begin{proof}
Using the argument in the proof of Theorem~\ref{thm:MD1best}, it is sufficient
    to show that for exponentially distributed service times, the deterministic
    preparation/arrival time maximizes $\mathbb{E}_{W\sim \pi}\left[\exp(-s
    W)\right]$ for any $s>0$.

The following two lemmas complete the proof of this theorem.

\begin{lemma}
\label{lem:smallestSigmaBest}
Among all arrival/preparation distributions with mean $\lambda^{-1}$
    (>$\mu^{-1}$), $\mathbb{E}_{W\sim \pi}\left[\exp(-s W)\right]$ for any
    $s>0$ is maximized by that arrival/preparation distribution for which the
    solution to the G/M/1 fixed point equation \[\sigma =
    \mathbb{E}_{A}\left[\exp\left(-(\mu-\mu~\sigma)~A\right) \right] \] is the
    smallest. 
\end{lemma}

\begin{lemma}
\label{lem:smallestSigmaDet}
Among all arrival/preparation distributions with mean $\lambda^{-1}$
    (>$\mu^{-1}$), the solution to the G/M/1 fixed point equation 
\[\sigma = \mathbb{E}_{A}\left[\exp\left(-(\mu-\mu~\sigma)~A\right) \right] \]
is the smallest for the deterministic arrival/preparation time $\lambda^{-1}$.
\end{lemma}
\end{proof}

\begin{proof}[Proof of Lemma~\ref{lem:smallestSigmaBest}]
The waiting time in a G/M/1 queue is exponentially distributed with mean
    $\frac{1}{\mu(1-\sigma)}$, where $\sigma$ is the solution to the fixed
    point equation
\[\sigma = \mathbb{E}_{A}\left[\exp\left(-(\mu-\mu~\sigma)~A\right) \right].\]
For exponentially distributed $W$, $\mathbb{E}\left[\exp(-s W)\right]$
    decreases with $\mathbb{E}[W]$. Hence, for a given $\mu$,
    $\mathbb{E}\left[\exp(-s W)\right]$ increases as $\sigma$ decreases, which,
    in turn, implies Lemma~\ref{lem:smallestSigmaBest}.
\end{proof}

Proof of Lemma~\ref{lem:smallestSigmaDet} is similar to the proof of Proposition~2 in \cite{ChatterjeeSeoEA17}.

\section{Conclusion and Outlook}
\label{sec:cncl}

Understanding the classical capacity of a quantum channel and the means by
which it can be achieved are fundamental issues in quantum information theory.
We derived an explicit \textcolor{black}{capacity-achieving} non-entangled projective measurement
strategy for i.i.d unital qubit channel. This implies that the classical
capacity of a unital qubit channel can be achieved without entanglement using
essentially classical resources.

Building on this insight, we showed that non-entangled projective measurements
achieve the classical capacity of a broad class of unital qubit queue-channels
that includes the well known unital qubit queue-channels like Pauli channels
and symmetric generalized amplitude damping channels.  In the special case of \textcolor{black}{the}
symmetric generalized amplitude damping channel, we show that our result on
unital qubit channel allows one to pick the capacity achieving product
encoding-decoding strategy (induced channel) out of a few natural yet
sub-optimal choices.

By taking the symmetric generalized amplitude damping channel as an example, we
demonstrate that ignoring the effect of decoherence in the buffer can lead to
an erroneous design choice. On the other hand, \textcolor{black}{a} queue-channel
based analysis, which offers a succinct model for decoherence in the buffer,
gives a procedure for finding the optimal operating point.

For operating a practical quantum communication system close to its capacity, efficient error correcting codes are essential. Our results from Sec.~\ref{UnitalQueue} imply that any capacity achieving classical error correcting code for binary symmetric channels, e.g., polar code, achieves the classical capacity of i.i.d. unital qubit channels. They also imply that a capacity achieving code for classical binary symmetric queue-channels achieves the classical capacity of unital qubit queue-channels when used in conjunction with the proposed product (classical to quantum) encoder and decoder. However, though the existence of a capacity achieving code for classical binary symmetric queue-channels is known \cite{ChatterjeeSeoEA17,MandayamJagannathanEA20}, the question of explicitly and efficiently finding such a code remains open.

Another important question follows from our work:  
can we construct induced classical channels for non-unital quantum channels with additive Holevo
information?  Obtaining such capacity achieving constructions remains an
interesting open problem.  To solve such a problem, one may follow the method
in this work. To use this method, one starts with a quantum channel with
additive Holevo information and then constructs an explicit induced channel
which achieves this Holevo information. As demonstrated in Sec.~\ref{SAmpDamp}
using the GADC, induced channels of this type can be non-trivial to construct.
For instance, in the case of non-unital GADC channels with additive Holevo
information, finding such induced channels remains an open problem. In
addition, finding the full parameter region where the GADC has additive Holevo
information also remains open.

Insights obtained from pursuing such open problems have the potential to not only
enrich the i.i.d setting with point-to-point quantum channels but also provide
a path to study non-i.i.d queue channel settings that arise in quantum
networks. Another challenging avenue for future work is to characterise the
queue channel capacity when the underlying noise model is not additive, as
could be the case for certain parameter ranges of the GADC. This may require a
fundamentally new approach to study quantum communication networks.

\section*{Acknowledgments}
The authors thank Mark M.~Wilde for useful comments on a previous draft of this
work.
VS gratefully acknowledges support from NSF CAREER Award CCF 1652560 and NSF
grant PHY 1915407. 
The work of AC was supported in part by the Department of Science and
Technology, Government of India under Grant SERB/SRG/2019/001809 and Grant
INSPIRE/04/2016/001171.
PM and KJ acknowledge the Metro Area Quantum Access Network (MAQAN) project,
supported by the Ministry of Electronics and Information Technology, India vide
sanction number 13(33)/2020-CC\&BT.

%\newpage
%

\end{document}